\documentclass[review,onefignum,onetabnum]{siamonline190516}

\usepackage[utf8]{inputenc}
\usepackage{algorithmic}
\usepackage{float}
\usepackage{amsfonts}
\usepackage{tikz}
\usetikzlibrary{positioning}
\usepackage{multicol}
\usepackage[caption=false]{subfig}
\usepackage{url}

\usepackage{breakurl}

\title{Bayesian Fusion of Data Partitioned Particle Estimates\thanks{Submitted to the editors October 26, 2020.\funding{This work was performed under the auspices of the U.S. Department of Energy by Lawrence Livermore National Laboratory under Contract DE-AC52-07NA27344. Funding for this work was provided by LLNL Laboratory Directed Research and Development grant 19-SI-004.}}}

\author{Caleb Miller\thanks{Department of Applied Mathematics, CU, Boulder, CO (\email{caleb.miller@colorado.edu, corcoran@colorado.edu}).} \and Michael D. Schneider\thanks{Lawrence Livermore National Lab (\email{schneider42@llnl.gov,bernstein8@llnl.gov}).}\and Jem N. Corcoran\footnotemark[2]\and Jason Bernstein\footnotemark[3]}

\date{October 2020}

\begin{document}

\maketitle

\begin{abstract}
     We present a Bayesian data fusion method to approximate a posterior distribution from an ensemble of particle estimates that only have access to subsets of the data. Our approach relies on approximate probabilistic inference of model parameters through Monte Carlo methods, followed by an update and resample scheme related to multiple importance sampling to combine information from the initial estimates. We show the method is convergent in the particle limit and directly suited to application on multi-sensor data fusion problems by demonstrating efficacy on a multi-sensor Keplerian orbit determination problem and a bearings-only tracking problem.
\end{abstract}



\section{Introduction}

Data fusion is the process through which different sources of information are combined to form a joint estimate of a process, target, or distribution of interest. There are many methods for data fusion, especially in the areas of multi-sensor measurements and target tracking~\cite{bar2011tracking,mahler2007statistical}. Several data fusion techniques exist to combine analytic descriptions or numerical estimates of probability density functions (PDFs)~\cite{genest1986combining}. Optimal Bayesian fusion of PDFs~\cite{dezert2001obfr} follows Bayes' rule for conditional density inference in  multiplying individual PDFs followed by division by the joint marginal distribution of the data. The normalised weighted geometric mean rule~\cite{bailey2012conservative}, an extension to non-Gaussian PDFs of the covariance intersection rule~\cite{julier1997non}, avoids division by the marginal data density by direct computation of numerical weights on each function involved in the product of individual measurement PDFs. 

To be applicable to non-linear and non-Gaussian problems, many methods must approximate distributions through sampling techniques. Sampling methods complicate the aforementioned fusion strategies as additional steps of estimation are then required to make PDF multiplication feasible. If the sampling-based approximation of the PDFs are Dirac mixture measures, multiplication must be applied after a process such as kernel density estimation to return non-zero estimates of the PDF products~\cite{mccabe2018fusion}. We present \emph{cross-pollination}, an alternative Bayesian sampling-based strategy, that uses multiple importance sampling ideas. Schneider et al.~\cite{schneider2016synthesis} adapted the multiple importance sampling work of Elvira et al.~\cite{elvira2019generalized}, to a sensor data fusion task through a multiple importance sampling framework. We improve this method via guarantees of convergence in the particle limit and validation of consistency of different weighting schemes, thus providing a data fusion strategy to accurately quantify uncertainty while avoiding the need for multiplication of PDF approximations by working directly with the existing particle estimates.

\section{Method}
\label{sec:MA}

\subsection{Objective, Set-Up, and Requirements}
The goal of cross-pollination is to approximate a posterior distribution $P(\boldsymbol{\theta}\mid D)$, conditioned on a set of data or observations $D$, from an ensemble of existing particle estimates of marginal posterior distributions  $P(\boldsymbol{\theta}\mid D_j)$, where the $D_j\, (j=1,\ldots,M)$ form a partition of $D$.

The random variable of interest $\boldsymbol{\theta}$ can take various forms e.g., an unknown value as in  \cref{sec:GAM}, a set of parameters which drives a deterministic motion model as in  \cref{sec:OD}, or a time-indexed collection of states which estimate a full trajectory as in  \cref{sec: SSB}. The initial particle estimates can be obtained online through importance sampling-resampling, or be pre-proccessed samples that contain the information from $D_j$ through other Monte Carlo or Markov chain Monte Carlo techniques. Cross-pollination can be applied in batches or sequentially, through repeated applications of the main steps of our weighting scheme and resampling. 

There are a few requirements to apply cross-pollination. The first requirement is that a common prior is used across all the estimates of the random variable of interest. The second is that likelihood functions must be available for all observations. These requirements are essential for Bayesian calculations and are included here for completeness. The third requirement is that the measurements must be independent, this is a common assumption in sampling-resampling frameworks as it allows for joint likelihood functions to be expressed as products of individual likelihood functions~\cite{stone2013bayesian}. The fourth requirement, which we will refer to as the \emph{core requirement}, is that particles obtained from any of the estimates must be able to map into the domain of the likelihood functions of the complement observations. This requirement is what allows our weighting scheme to occur. We will see that in the inference of parameters in a deterministic motion model \cref{sec:OD} the core requirement is easily satisfied, while in the case of a stochastic motion model more work is required to satisfy the core requirement \cref{sec: SSB}.

\subsection{Measure Theoretic Description of Method}
In the discussion of the method that follows we will assert the data $D = (z_1,z_2,\ldots,z_M)$ was partitioned into single observations, $D_j=z_{j}$, and that the initial samples were obtained through importance sampling. Both of these choices are only here for the purpose of clarity; we will divert from this set up in the applications in \cref{sec:EG}. To speak accurately about updating the samples we turn to a measure theoretic viewpoint.

Let $\pi$ be the probability measure corresponding to the posterior distribution $P(\boldsymbol{\theta} \mid D)$, $\pi_{j,0}$ be the probability measure corresponding to the initial marginal posterior distribution $P(\boldsymbol{\theta} \mid D_j)$. We denote the probability measures that estimate these measures as $\pi^N$ and $\pi^N_{j,0}$ respectively, noting that the superscript $N$ indicates that these estimates are each comprised of $N$ particles approximately sampled from their respective distributions.

Let $\boldsymbol{\theta}$ be a random variable of interest with prior $P(\boldsymbol{\theta})$ and observations $D := (z_1,z_2,\ldots,z_M)$ with measurement uncertainty which is encapsulated by a likelihood function. For each member of the partitioned data $D_j$, a Dirac mixture measure $\pi_{j,0}^N$ is formed from importance sampling using a likelihood function associated with the observation, $g_j(\boldsymbol{\theta})$, where $P\left(D_j\mid \boldsymbol{\theta}\right)\propto g_j(\boldsymbol{\theta})$. Specifically,
\begin{equation}
\label{eq:pj0}
\pi_{j,0}^N = \frac{1}{N}\sum_{i=1}^N\delta_{\boldsymbol{\theta}_j^{(i)}} \quad (j = 1,\ldots,M), 
\end{equation}
where the particles $\boldsymbol{\theta}_j^{(i)}\, (i=1,\ldots,N)$ are obtained through importance sampling with the common prior $P(\boldsymbol{\theta})$ as the importance distribution and the likelihood function of the observation used to weight and resample. We denote the Dirac function $\delta(x-x^*)$ as $\delta_{x^*}$. We are to weight the particles of our initial measure $\pi_{j,0}^N$, which is built from a single observation, to form an empirical measure $\pi_{j}^N$ that estimates $\pi$. We use the Radon-Nikodym derivative to obtain the weights needed to match expectations between the two measures. As $\pi$ is absolutely continuous with respect to $\pi_{j,0}$ we have that
\begin{equation}
\label{eq:rn}
\frac{d\pi}{d\pi_{j,0}}(\boldsymbol{\theta}) = \frac{P(D_{-j}\mid \boldsymbol{\theta})}{P(D_{-j}\mid D_j)} \propto \prod_{k\neq j}g_k\left( \boldsymbol{\theta}\right),
\end{equation}
where $D_{-j} = (z_1,z_2,\ldots,z_{j-1},z_{j+1},\ldots,z_M)$ and the final product is due to the independence of the measurements. Each particle is now updated with a weight that corresponds to the value in the product of likelihood functions of the complementary data. We refer to this step as \emph{cross-epoch weighting}.

\begin{equation}
\label{eq:pjn}
\pi_{j}^N = \sum_{i=1}^N w_{j}^{(i)}\delta_{\boldsymbol{\theta}_{j}^{(i)}},\quad w_{j}^{(i)} = \frac{\widetilde{w}_{j}^{(i)}}{\sum_{i'=1}^N \widetilde{w}_{j}^{(i')}},\quad \widetilde{w}_{j}^{(i)} = \prod_{k\neq j}g_k\left(\boldsymbol{\theta}_j^{(i)}\right).
\end{equation}
Notice that the structure of the weights is reflected in an application of Bayes' rule working directly with the probability distribution functions,
\begin{equation}
    P(\boldsymbol{\theta}\mid D) = P(\boldsymbol{\theta}\mid D_{-j},D_j) = \frac{P(D_{-j}\mid \boldsymbol{\theta},D_j)}{P(D_{-j}\mid D_j)}P(\boldsymbol{\theta}\mid D_j)=\frac{P(D_{-j}\mid \boldsymbol{\theta})}{P(D_{-j}\mid D_j)}P(\boldsymbol{\theta}\mid D_j).
\end{equation}
Now, we combine the empirical measures through pooling and resampling to obtain $N$ equally weighted particles. One way to achieve this is through combining all of the measures with a summation and dividing by the total number of measures

\begin{equation}
\label{eq:pjmn}
\pi^{MN} = \frac{1}{M} \sum_{j=1}^M \pi_j^N = \frac{1}{M} \sum_{j=1}^M \sum_{i=1}^N w_{j}^{(i)}\delta_{\boldsymbol{\theta}_{j}^{(i)}},
\end{equation}
and then employing a resampling technique, say multinomial, to obtain our final approximation
\begin{equation}
\label{eq:pn}
\pi^N = \frac{1}{N}\sum_{i=1}^N \delta_{\boldsymbol{\theta}^{(i)}}, 
\end{equation}
where $\boldsymbol{\theta}^{(i)}\, (i=1,\ldots,N)$ are identically and independently distributed from the measure $\pi^{MN}$. 

We can fuse the particles in another way; instead of normalizing the weights by data partition, as in \eqref{eq:pjn}, we can go straight to a fused measure by normalizing all the weights together,
\begin{equation}
\label{eq:pmnt}
    \pi^{MN}_{t} = \sum_{j=1}^M\sum_{i=1}^N w_{j}^{(i)}\delta_{\boldsymbol{\theta}_{j}^{(i)}},\quad w_{j}^{(i)} = \frac{\widetilde{w}_{j}^{(i)}}{\sum_{j'=1}^M\sum_{i'=1}^N \widetilde{w}_{j'}^{(i')}},\quad \widetilde{w}_{j}^{(i)} = \prod_{k\neq j}g_k\left(\boldsymbol{\theta}_j^{(i)}\right).
\end{equation}
Then, in the same fashion we would employ a sampling technique to get to our final estimate of $N$ particles,
\begin{equation}
\label{eq:pnt}
\pi^N_{t}= \frac{1}{N}\sum_{i=1}^N \delta_{\boldsymbol{\theta}^{(i)}}, 
\end{equation}
where $\boldsymbol{\theta}^{(i)}\, (i=1,\ldots,N)$ are identically and independently distributed from the measure $\pi_t^{MN}$.
We refer to the methods as `norming-apart' and `norming-together', respectively. This norming-apart process is diagrammed in \cref{diagram} both methods are proved to converge under a root mean square distance in the particle limit when the initial estimates are formed using importance sampling. The proof of the norming-apart method is presented in \cref{sec:theory}. The ease of these formulations suggest straightforward algorithms for both methods, pseudo-code is provided in \cref{sec:code}.


\begin{figure}[tbhp]
\hspace*{3cm}%
\begin{tikzpicture}[
datanode/.style={circle, draw=blue, thin, minimum size=7mm},
pnode/.style={rectangle, draw=blue, thin, minimum size=5mm},
lnode/.style={rectangle, draw=red, thin, minimum size=5mm},
onode/.style={rectangle, draw=red!0, thin, minimum size=8mm},
]

\pgfmathsetseed{4}
   
\foreach \i in {1,...,25}
   \draw[xshift=random()*1.5cm, yshift=random()*1.5cm, fill=cyan!50] (0,0) circle (1mm);
   
\foreach \i in {1,...,25}
   \draw[xshift=random()*1.5cm, yshift=random()*1.5cm, fill=blue!50] (3,0) circle (1mm);
   
\foreach \i in {1,...,25}
   \draw[xshift=random()*1.5cm, yshift=random()*1.5cm, fill=magenta!50] (6,0) circle (1mm);

\pgfmathsetseed{4}  

\foreach \dia in {1mm,.1mm,1mm,1mm,1mm,.5mm,.1mm,.3mm,.3mm,.3mm,1mm,.1mm,1mm,1mm,.5mm,.7mm,1mm,.3mm,1mm,1mm,1mm,1mm,.7mm,1mm}
    \draw[xshift=random()*1.5cm, yshift=random()*1.5cm, fill=cyan!50] (0,-3) circle (\dia);

\foreach \dia in {.3mm,1mm,1mm,1mm,.7mm,.3mm,.5mm,.5mm,1mm,.3mm,1mm,.5mm,1mm,.7mm,.1mm,1mm,.7mm,1mm,.3mm,.7mm,.3mm,1mm,.5mm,.3mm}
   \draw[xshift=random()*1.5cm, yshift=random()*1.5cm, fill=blue!50] (3,-3) circle (\dia);
   
\foreach \dia in {.3mm,.3mm,.3mm,1mm,1mm,1mm,1mm,.7mm,.5mm,.3mm,.5mm,.7mm,.5mm,1mm,.7mm,.3mm,1mm,1mm,1mm,1mm,.5mm,.3mm,1mm,.5mm}
   \draw[xshift=random()*1.5cm, yshift=random()*1.5cm, fill=magenta!50] (6,-3) circle (\dia);

\pgfmathsetseed{4}  

\foreach \dia in {1mm,.1mm,1mm,1mm,1mm,.5mm,.1mm,.3mm,.3mm,.3mm,1mm,.1mm,1mm,1mm,.5mm,.7mm,1mm,.3mm,1mm,1mm,1mm,1mm,.7mm,1mm}
    \draw[xshift=random()*1.5cm, yshift=random()*1.5cm, fill=cyan!50] (3,-6) circle (\dia);

\foreach \dia in {.3mm,1mm,1mm,1mm,.7mm,.3mm,.5mm,.5mm,1mm,.3mm,1mm,.5mm,1mm,.7mm,.1mm,1mm,.7mm,1mm,.3mm,.7mm,.3mm,1mm,.5mm,.3mm}
   \draw[xshift=random()*1.5cm, yshift=random()*1.5cm, fill=blue!50] (3,-6) circle (\dia);
   
\foreach \dia in {.3mm,.3mm,.3mm,1mm,1mm,1mm,1mm,.7mm,.5mm,.3mm,.5mm,.7mm,.5mm,1mm,.7mm,.3mm,1mm,1mm,1mm,1mm,.5mm,.3mm,1mm,.5mm}
   \draw[xshift=random()*1.5cm, yshift=random()*1.5cm, fill=magenta!50] (3,-6) circle (\dia);
   
\pgfmathsetseed{4}  

\foreach \dia in {1mm,0mm,1mm,1mm,1mm,0mm,0mm,0mm,0mm,0mm,1mm,0mm,1mm,1mm,0mm,1mm,1mm,0mm,1mm,1mm,1mm,1mm,0mm,0mm}
    \draw[xshift=random()*1.5cm, yshift=random()*1.5cm, fill=blue!25] (3,-9) circle (\dia);

\foreach \dia in {0mm,1mm,1mm,1mm,1mm,0mm,0mm,0mm,1mm,0mm,0mm,0mm,0mm,0mm,0mm,1mm,0mm,1mm,0mm,0mm,0mm,1mm,1mm,0mm}
   \draw[xshift=random()*1.5cm, yshift=random()*1.5cm, fill=blue!25] (3,-9) circle (\dia);
   
\foreach \dia in {0mm,0mm,0mm,1mm,1mm,1mm,1mm,1mm,1mm,0mm,0mm,1mm,0mm,0mm,1mm,0mm,1mm,1mm,1mm,1mm,0mm,0mm,0mm,0mm}
   \draw[xshift=random()*1.5cm, yshift=random()*1.5cm, fill=blue!25] (3,-9) circle (\dia);
 
\node at (-1, .75)    (a) {$\pi^N_{j,0}$}; 
\node at (-1, .75-3)  (b) {$\pi^N_{j}$}; 
\node at (-1, .75-6)  (c) {$\pi^{MN}$}; 
\node at (-1, .75-9)  (d) {$\pi^{N}$}; 
\draw[thick,->] (a.south) -- (b.north) node[midway,sloped,below] {\small Weighted};
\draw[thick,->] (b.south) -- (c.north) node[midway,sloped,below] {\small Pooled};
\draw[thick,->] (c.south) -- (d.north) node[midway,sloped,below] {\small Resampled};

\draw[rounded corners,cyan!50,thick] (-.15-.02,-3.25-.02+3) rectangle ++(2,2);
\draw[rounded corners,blue!50,thick] (-.15-.02+3,-3.25-.02+3) rectangle ++(2,2);
\draw[rounded corners,magenta!50,thick] (-.15-.02+6,-3.25-.02+3) rectangle ++(2,2);
\draw[rounded corners,magenta!50,thick] (-.15-.02,-3.25-.02) rectangle ++(2,2);
\draw[rounded corners,blue!50,thick] (-.15+.04,-3.25+.04) rectangle ++(1.875,1.875);
\draw[rounded corners,magenta!50,thick] (-.15-.02+3,-3.25-.02) rectangle ++(2,2);
\draw[rounded corners,cyan!50,thick] (-.15+.04+3,-3.25+.04) rectangle ++(1.875,1.875);
\draw[rounded corners,blue!50,thick] (-.15-.02+6,-3.25-.02) rectangle ++(2,2);
\draw[rounded corners,cyan!50,thick] (-.15+.04+6,-3.25+.04) rectangle ++(1.875,1.875);

\end{tikzpicture}
\label{diagram}
\caption{Cross-Pollination Process: Particles (circles) have initially only gained the statistical information (squares) of their own data (shared color between circle and square). The process continues by imparting the statistical information that the other particles have not been exposed to yet (different color between circles and squares) and imparting weights (size of particle, larger size indicating larger weight). The particles are then pooled and resampled based on their weights.}
\end{figure}
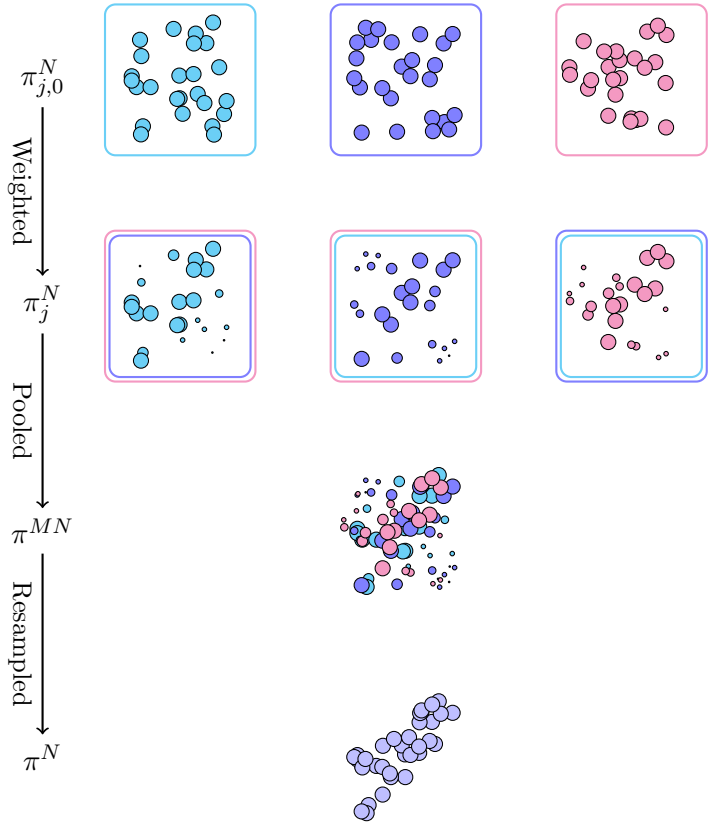

\subsection{Relation to Multiple Importance Sampling}

Importance sampling is a well-known technique of sampling from a known distribution $q(\boldsymbol{\theta})$ and adjusting the samples with importance weights to approximate a target distribution $P(\boldsymbol{\theta})$~\cite{robert2013monte}. The typical form of the importance weights for a sample $\boldsymbol{\theta}^{(i)}$ is $w\left(\boldsymbol{\theta}^{(i)}\right) = P\left(\boldsymbol{\theta}^{(i)}\right)/q\left(\boldsymbol{\theta}^{(i)}\right)$. Multiple importance sampling expands this idea by introducing a number of proposal distributions $q_j(\boldsymbol{\theta})\, (j=1,\ldots,M)$ that can be adjusted in various ways to estimate the desired distribution. Both the weighting and ordering of selection from proposal distributions can affect the final estimate~\cite{elvira2019generalized}. 

The connection to cross-pollination is formed by considering the importance sampling target distribution as the posterior conditioned on all the data $P(\boldsymbol{\theta}) = P(\boldsymbol{\theta}\mid D)$ and the importance sampling proposal distributions to be the marginal posteriors $q_j(\boldsymbol{\theta}) = P(\boldsymbol{\theta}\mid D_j)$, which makes the initial measures $\pi^N_{j,0}$ the collections of samples obtained from each of our proposal distributions. By using the data-dependent proposals we have created a data fusion algorithm that can make use of existing multiple importance sampling theory. Both the norming-apart and norming-together cross-pollination schemes resemble aspects of \emph{generic diverse population Monte Carlo} methods discussed by Elvira et. al~\cite{elvira2017population}. Schneider's original work on cross-pollination~\cite{schneider2016synthesis} used the \emph{deterministic mixture weights} proposed by Elvira et al.~\cite{elvira2019generalized}. In this case, one would have,
\begin{equation}
\label{pmndm}
    \pi^{MN}_{dm} = \sum_{j=1}^M\sum_{i=1}^N w_{j}^{(i)}\delta_{\boldsymbol{\theta}_{j}^{(i)}},\quad w_{j}^{(i)} = \frac{\widetilde{w}_{j}^{(i)}}{\sum_{j'=1}^M\sum_{i'=1}^N \widetilde{w}_{j'}^{(i')}},\quad \widetilde{w}_{j}^{(i)} = \frac{\prod_{k=1}^M g_k\left(\boldsymbol{\theta}_{j}^{(i)}\right)}{\sum_{k=1}^Mg_k\left(\boldsymbol{\theta}_{j}^{(i)}\right)},
\end{equation}
and, through sampling would obtain,
\begin{equation}
\label{eq:pndm}
\pi^N_{dm} = \frac{1}{N}\sum_{i=1}^N \delta_{\boldsymbol{\theta}^{(i)}}, 
\end{equation}
where $\boldsymbol{\theta}^{(i)}\, (i=1,\ldots,N)$ are identically and independently distributed from the measure $\pi_{dm}^{MN}$.
This weight formulation is of particular interest as Elvira et al.~\cite{elvira2019generalized} catalog it as the scheme having the minimum variance of the multiple importance sampling schemes they examined. However, this scheme is only shown to be consistent, in the limit of the number of proposals, when the normalizing constant is known. This makes the deterministic mixture weights scheme less applicable to the problem at hand which demands a fixed number of proposal distributions depending on how the data has been partitioned and has no knowledge of the normalizing constant, which is typically estimated by the sum of the weights in importance sampling theory.

\section{Applications}
\label{sec:EG}

\subsection{Set-up}
We demonstrate the algorithms on three examples: the first is an example with an analytic posterior distribution that serves as an introduction to application of the method as well as a verification of the convergence rate, the second is an orbit determination problem where particle estimates borne of data from two different sensors are fused via the norming-together method, the third example is a bearings-only tracking example that shows one way the ideas of cross-pollination may be carried into a sequential problem with a stochastic motion model. For simplicity in the latter two examples, we consider normal or multivariate normal likelihood functions for measurements. We denote the PDF of a normal distribution with mean $\mu$ and variance $\sigma^2$ evaluated at $x$ as $\mathcal{N}(x;\mu,\sigma^2)$. Similarly, for a multivariate normal we will have $\mathcal{MVN}(\mathbf{x};\boldsymbol{\mu},\boldsymbol{\Sigma})$ to be the PDF of a multivariate normal with mean $\boldsymbol{\mu}$ and covariance matrix $\boldsymbol{\Sigma}$ evaluated at $\mathbf{x}$. 

\subsection{Convergence Verification on Gamma Example}
\label{sec:GAM}
Our theory argues for a convergence rate of $\sqrt{N}$ of our estimator $\pi^N$ to the probability measure $\pi$ that corresponds to the posterior distribution. We verify on an example with Gamma distributions. We denote that $x$ is distributed from a Gamma distribution with shape parameter $k_0$ and rate parameter $\theta$ by $x\sim \mathcal{G}(k_0,\theta)$. First we give our random variable of interest a gamma prior distribution $x \sim \mathcal{G}(k_0,\theta_0)$. Next we have three observations each with a Gamma distribution so that $y_j|x \sim \mathcal{G}(k_j,x^{-1})\, (j=1,2,3)$. This set-up admits an analytic posterior distribution,

\begin{equation}
   x|y_1,y_2,y_3 \sim \mathcal{G}\left(k_0+k_1+k_2+k_3, (\theta_0^{-1} +y_1+y_2+y_3)^{-1}\right).
\end{equation}

For this experiment we chose $k_0 = 5/2,\,\theta_0=1/2,\,y_1=1,\,y_2=2,\,y_3=3,\,k_1=4,\,k_2=10,\,k_3=25$. The sets of initial particles corresponding to $\pi_{j,0}^N\,(j=1,2,3)$ were obtained by one round of standard importance sampling. The particles were then cross-epoch weighted and resampled under both methods of norming-together and norming-apart cross-pollination. We conducted $1000$ Monte Carlo trials for $N= 10^2,\,10^3,\,10^4,\,10^5,\,10^6$ particles.

\begin{figure}[tbhp]
\centering
\subfloat[]{\label{fig:a}\includegraphics[scale=0.4]{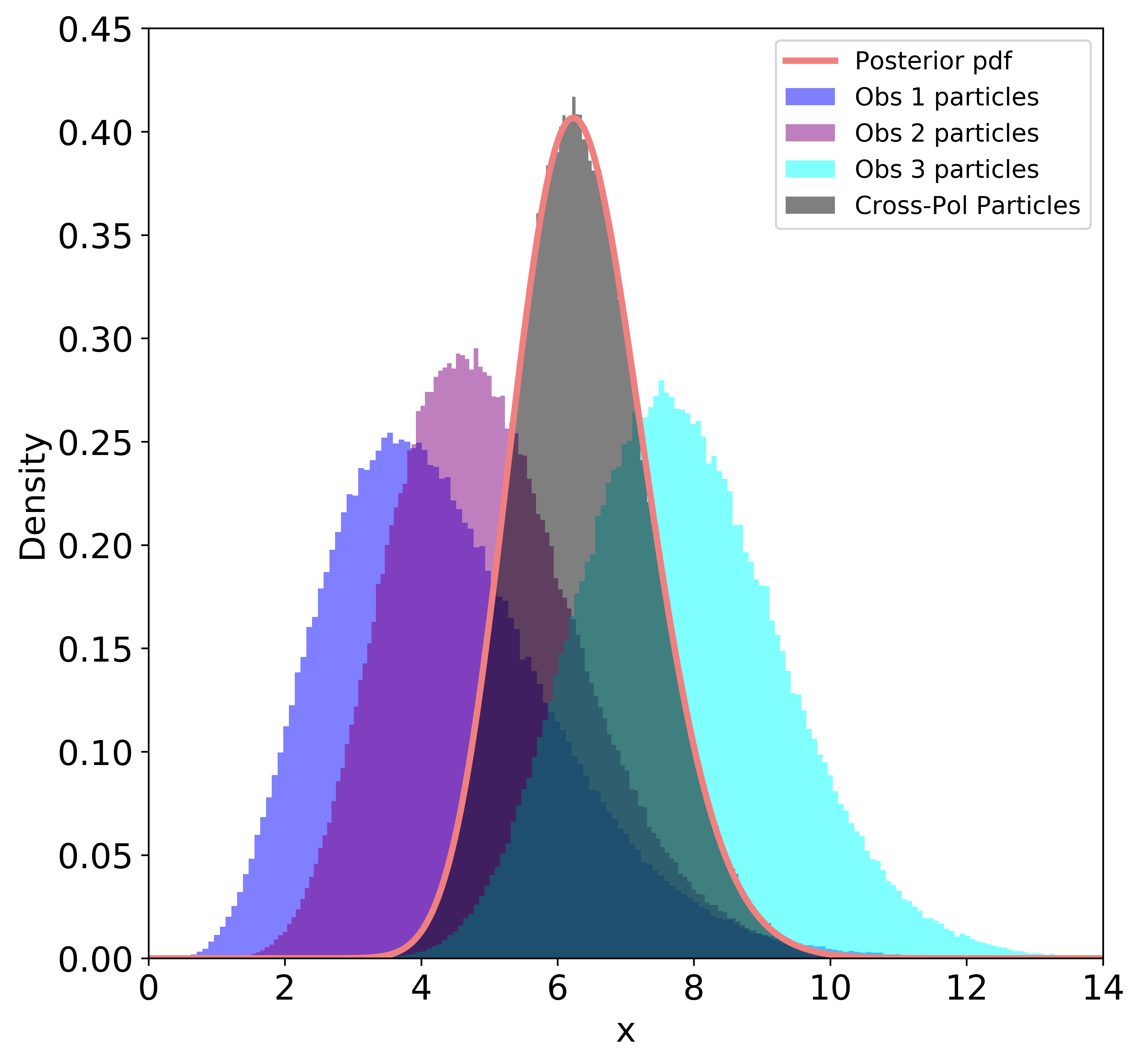}}
\hfill
\subfloat[]{\label{fig:b}\includegraphics[scale=0.4]{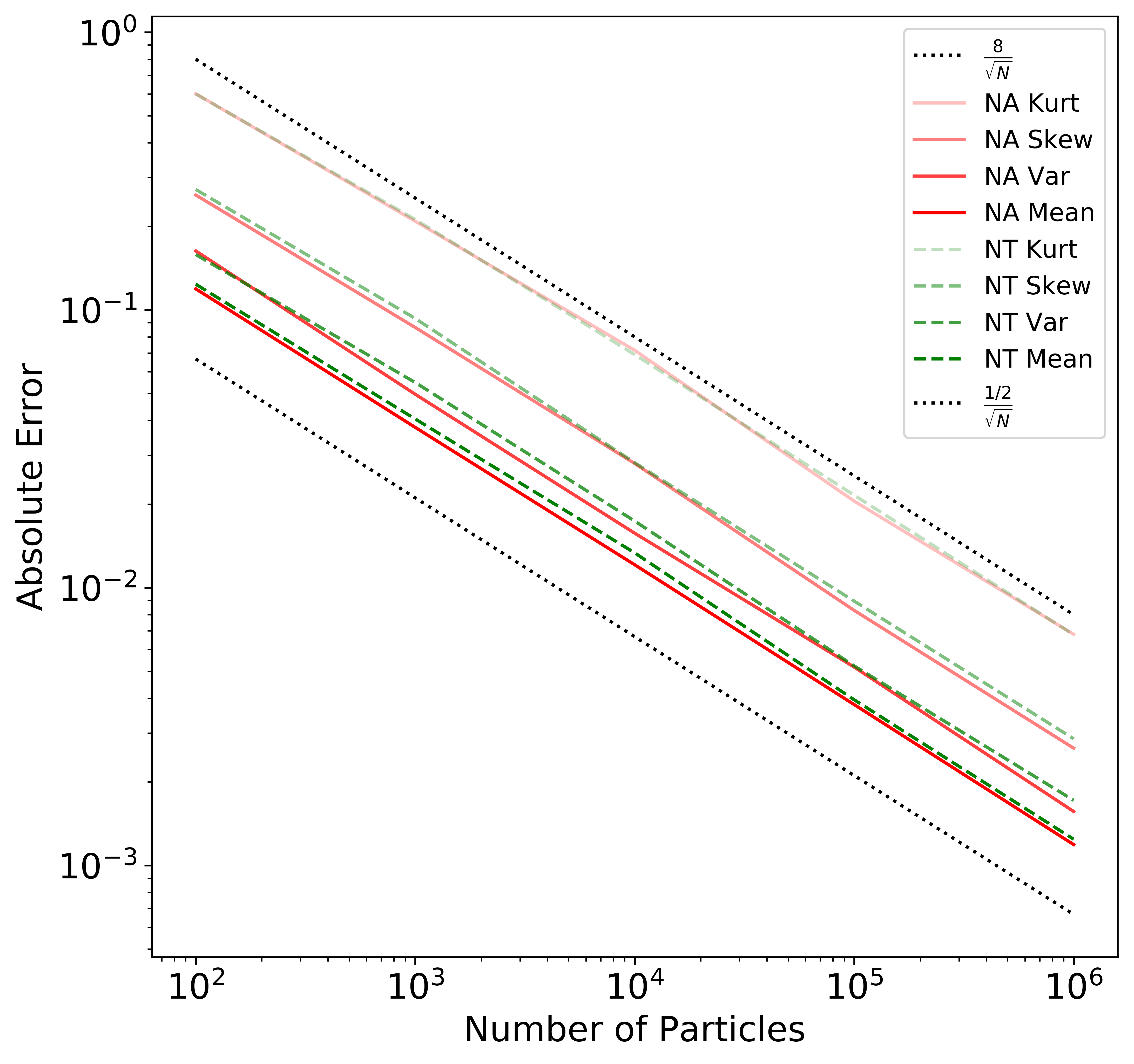}}
\caption{(a)  Blue, purple, and cyan histograms show the initial sets of particles for a trial with $10^6$ particles taken created from observations $1,\,2\,,\text{and }3$ respectively. The grey histogram is of particles that have been selected after a single run of cross-pollination. Solid pink line is the posterior PDF. (b) The black dotted lines follow $\frac{8}{\sqrt{N}}$ (top) and $\frac{1/2}{\sqrt{N}}$ (bottom). Solid red lines from darkest to lightest plot the mean over $1000$ Monte Carlo trials of the absolute error of the estimated first four moments (mean, variance, skew, kurtosis) against the true value of the moments for the norming-apart method. Similarly the green dashed lines describe the same situation for the norming-together method.}
\label{fig:testfig}
\end{figure}

To reiterate, we formed the particles from the observations (labelled Obs 1 particles etc in \cref{fig:a}) based on drawing particles from the prior and importance sampling with the single Gamma likelihood. The particles were then put through the cross-pollination process of cross-epoch weighting and resampling to form the final estimate of particles (labelled Cross-pol particles in \cref{fig:a}). Our theory is supported by \cref{fig:b} as we see the mean absolute error decreases at a rate proportional to the square root of the number of particles for each moment. The proportion changes depending on which moment is being estimated, as one would expect the proportionality constant is higher for the higher order moments. 

We included this example to provide an introduction to the application of the method and check convergence rates where the posterior distribution had an analytic form. We believe cross-pollination can be a useful tool in multi-sensor data fusion applications as demonstrated in the following applications.

\subsection{Orbit Determination}
\label{sec:OD}

Space traffic management is becoming an increasingly challenging remote sensing and data fusion problem as the number of satellites and debris around the Earth continue to grow. A key aspect of space traffic management problems includes orbit determination. Multi-sensor data fusion for an orbital determination problem has been explored using the geometric mean density (GMD) fusion rule~\cite{mccabe2018fusion}.

In an orbit determination scenario where ground-based optical angles-only measurements are accessible online and orbital measurements are lagged, practitioners may desire making ground-based estimates immediately and fusing the orbital estimates when they become available at a later time. We simulate this situation below for Keplerian orbits.

Keplerian orbits trace closed ellipses in space and are described in a six-dimensional state space, which is typically described by the six orbital elements: semi-major axis, eccentricity, inclination, argument of periapsis, right ascenscion of the ascending node, and the true anomaly. Equivalently, these orbits can be identified by a three-dimensional position and a three-dimensional velocity at a given time. We let $\boldsymbol{\theta} = (x,y,z,\dot{x},\dot{y},\dot{z})$ at time $t_0$, $P(\boldsymbol{\theta})$ be the prior, and consider two sets of three right ascension (RA) and declination (DEC) measurements occurring: one taken from a sensor in low Earth orbit (LEO) $D_1$ over three minutes and the other from a ground-based observer $D_2$ over a four minute period taking place six minutes after the final orbiting-sensor measurement. Sensor measurements are modeled as multivariate Gaussian with mean given by the true RADEC and standard deviations given by $\sigma_{\rm LEO} =2 \text{ arcseconds}$ and $\sigma_{\rm ground} =20 \text{ arcseconds}$. Therefore, the likelihood function for measurements $j$ of the LEO observer is,
\begin{equation}
    g_j\left(\boldsymbol{\theta}^{(i)}\right) = \mathcal{MVN}\left(H_j\left(\boldsymbol{\theta}^{(i)}\right);z_j,\boldsymbol{\Sigma}_{LEO}\right),
\end{equation}
where the non-linear measurement function $H_j$ maps the particles into the correct measurement space by propagating the orbit specified by $\boldsymbol{\theta}^{(i)}$ to the correct observation time $t_j$ and retrieving the RADEC measurements that would have been obtained from $jth$ sensor, and the covariance matrix is  $\boldsymbol{\Sigma}=\sigma_{\rm LEO}^2 \mathbf{I}_2$. Of course, the ground based likelihoods take the same form but with $\sigma_{\rm ground}$. This set-up satisfies the requirements of cross-pollination.

We will apply cross-pollination to the estimates that comprise $\pi^N_{1,0}$ and $\pi^N_{2,0}$ which consist of $50,000$ samples each, obtained with an MCMC sampler to produce initial orbit determinations based on the data provided. We fuse these samples using norming-together cross-pollination and perturb the samples to achieve our final estimate $\pi^N$. We diagram this process in \cref{fig:3-orbit}.

\begin{figure}[tbhp]
  \centering
  \label{fig:3-orbit}\includegraphics[scale=.39]{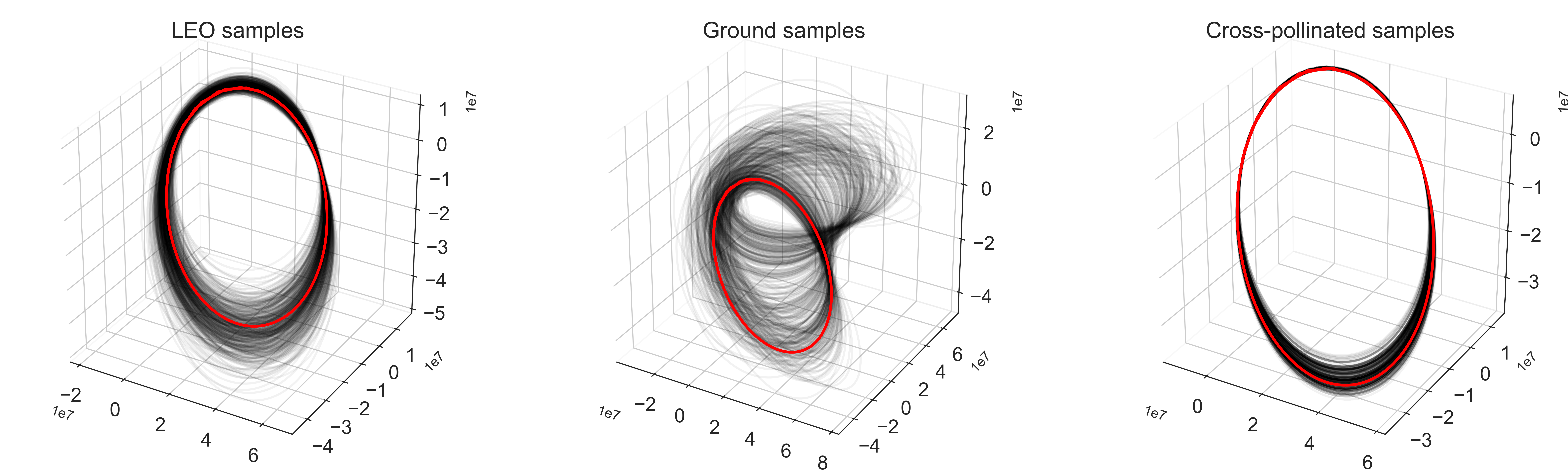}
  \caption{From left to right: a random sample of $450$ of the $50,000$ particles estimates traced out orbits (grey) and the true orbit (red) for the LEO sensor, ground-based sensor, and the estimates after being refined through the norming-together cross-pollination process.}
\end{figure}

From \cref{fig:3-orbit} we can see that the refined orbits from the norming-together cross-pollination process are providing an improved picture of the situation and uncertainty. For this example the root mean square error of the positional portion of the orbital elements for the samples before cross-pollination was $280\,km$ and $681\,km$ for the ground and LEO samples respectively, after cross-pollination this was improved to an root mean square error of $32\,km$. Similarly the root mean square errors of the the velocity portion improved from $104\,m/s$ and $2243\, m/s$ to $42\,m/s$. We can see in \cref{fig:3-orbit} a distinct improvement in the uncertainty as the ground sensor MCMC process allowed many samples from a degenerate orbital plane, due to a small observation window and large uncertainty in the measurements, that has been eliminated in the fused results. Furthermore, the distinct ``banana shaped uncertainty''~\cite{horwood2014gauss} that emerges in orbit determination problems is seen in the rightmost image of \cref{fig:3-orbit}, indicating uncertainty realism.

\subsection{Multiple Sensor Sequential Bearings}
\label{sec: SSB}

We have applied cross-pollination to two situations where the core requirement was easily satisfied\textemdash likelihood functions could be immediately applied to particles. This was due to fusing three measurements at a single epoch in \cref{sec:GAM} and utilizing a deterministic motion model in \cref{sec:OD}. Not all situations will follow these formats. In this application we demonstrate an adaptation of the cross-pollination ideas for a sequential smoothing problem with a stochastic motion model.

Say we desire a smoothing distribution $P(x_{1:6}\mid y_{1:6})$ for a tracking problem where an object follows stochastic dynamics. Furthermore, suppose that the estimates obtained are from processes ran on two sensors with differing cadences so that we desire to fuse the trajectories (particles) of $\pi_1^N$ corresponding to $P(x_1,x_3,x_5\mid D_1=(y_1,y_3,y_5))$ and $\pi_2^N$ corresponding to $P(x_2,x_4,x_6\mid D_2=(y_2,y_4,y_6))$. If the $y_i$ took place at distinct times $t_i$ and the times are increasing $t_1<t_2<t_3<\ldots$, it is clear that a likelihood that corresponds to $t_2$ can not be applied to a trajectory from $\pi^N_1$ as the state at time $t_2$ has not been recorded. We note that this corresponds to the problem of having to incorporate out-of-sequence measurements (OOSM) which has been discussed in tracking and data fusion literature~\cite{bar2004one}. An attractive solution to implement cross-pollination would be to consider the Bayesian updates used in the OOSM particle filter~\cite{orton2002bayesian}. In the following example we mitigate this problem by saving the state of the particles where any observation occurred or will occur. Depending on sensor architecture this may not be possible and other mentioned methods could be employed so that likelihood functions could be applied at all necessary times. 

We apply cross-pollination to a bearings-only tracking problem inspired by Gordon et al.~\cite{gordon1993novel}. Several other methodologies have been developed to perform well on the bearing-only tracking problem, for example the resample-move filter~\cite{gilks2001following}. We choose this problem not to compete with the other methods but to demonstrate the versatility of cross-pollination to a multi-sensor sequential problem with a stochastic motion model. We will be estimating the the state vectors of positions and velocities, up to time $t$, which will be driven through a linear motion model with additive Gaussian noise. That is,

\begin{equation}
    \boldsymbol{\theta}_t = (\mathbf{x_1},\ldots\mathbf{x_t}),\quad \mathbf{x_j} = \Phi \mathbf{x_{j-1}} +\Gamma \mathbf{w_j},
\end{equation}

\begin{equation}
    \mathbf{x_j} = (x,\dot{x},y,\dot{y})_j^T, \quad \mathbf{w_j} =(w_x,w_y)_j^T,\quad \mathbf{w_j} \sim N(\mathbf{0},\sigma_q^2\mathbf{I}_2) ,
\end{equation}

\begin{align}
  \Phi = \begin{bmatrix}1&1&0&0\\0&1&0&0\\0&0&1&1\\0&0&0&1\end{bmatrix} \quad \Gamma =\begin{bmatrix}0.5&0\\1&0\\0&0.5\\0&1\end{bmatrix}.
\end{align}
Bearing observations of the process will be obtained through two different sensors with additive Gaussian noise. The $jth$ sensor is located at $(x_{S_j},y_{S_j})$ (actual locations $(-1,1)$ and $(1,-1)$) so that likelihood functions of the measurements applied to particles are 
\begin{equation}
    g_j\left(\boldsymbol{\theta}^{(i)}\right) = \mathcal{N}\left(H\left(\boldsymbol{\theta}^{(i)}\right);z_j,\sigma_r^2\right), \quad H\left(\boldsymbol{\theta}^{(i)}\right) = \arctan\left(\frac{y_i-y_{Sj}}{x_i-x_{Sj}}\right).
\end{equation}

The standard deviations of the sensor and model noise are $\sigma_r = 0.005$ and $\sigma_q = 0.001$ respectively. The prior for initial state is a multivariate Gaussian such that
\begin{equation}
    \mathbf{x_0} \sim \mathcal{MVN}(\mathbf{x}_p,\boldsymbol{\Sigma}_p),\quad \mathbf{x}_p = [0,0,0.4,-0.05]^T, \quad \mathbf{\boldsymbol{\Sigma}_p} =\begin{bmatrix}0.5^2&0&0&0\\0&0.005^2&0&0\\0&0&0.3^2&0\\0&0&0&0.01^2\end{bmatrix}.
\end{equation}
The initial state is $\mathbf{x_{0}} = [-0.5,0.001,0.7,-0.55]^T$. This starting state is propagated through the motion model for $20$ time steps to form the true trajectory.

To demonstrate further flexibility of the ideas of cross-pollination we partition the data completely so that $D_{j}=y_{j}$ with $D = (y_1,y_2,\ldots,y_{20})$. We obtain the sample trajectories $\boldsymbol{\theta}_j^{(i)}$ that comprise the empirical measure $\pi^N_{j,0}$ corresponding to $P(\boldsymbol{\theta_j}\mid D_j=y_j)$ by propagating prior samples forward until time $j$, saving the state at each time ($1,2,\ldots,j$), and then performing importance sampling and resampling of the whole trajectory based on the $jth$ likelihood function. 

To start the process of sequential cross-pollination we begin with $\pi^N_{1,0}$ and motion the trajectories to $t_2$. The forward motioned trajectories of $\pi^N_{1,0}$ are then fused with the trajectories of $\pi^N_{2,0}$ by weighting the motioned trajectories with the likelihood $g_2$ and weighting the trajectories of $\pi^N_{2,0}$ with the ``backward in time'' likelihood $g_1$. After the cross-epoch weighting the trajectories can be pooled and resampled using norming-apart or norming-together processes as they have seen all the relevant statistical information up to the second epoch. To continue the process the cross-pollinated trajectories denoted $\pi_{2,*}^N$ are motioned forward to $t_3$ and then cross-pollinated with $\pi_{3,0}^N$. Notice that the trajectories of $\pi^N_{3,0}$ would now be weighted with two likelihoods $g_1$ and $g_2$. Extending this notion, the samples of $\pi_{j,0}^N$ would be weighted with a product of $j-1$ likelihoods. This process continues until the desired time, here $t=20$, so that $\pi_{20,*}^N$ would estimate the measure that corresponds to $P(\boldsymbol{\theta}_{20} \mid D = y_{1:20})$. Pseudo-code for this process is provided in \cref{sec:code}. Due to particle sparsity issues raised by the application of several backward in time likelihoods we implemented this process with log likelihoods and resampled the particles during the process of applying the cross-epoch weight whenever the effective sample size was below a threshold and perturbed the particles that survived the resampling. Notice that if the initial particles are obtained via importance sampling resampling and the resample step accepts only particles originating from the measure $\pi^N_{j,*}$ at each step then this algorithm is essentially the sequential importance resampling particle filter~\cite{gordon1993novel}. The additional steps provide new possible trajectories that could help reduce particle sparsity.

In \cref{fig:SB} we see wide swaths of initial particles that comprise empirical measures at particular epochs $\pi^N_{j,0} (j=1,4,7,10,13,16,19)$. The sequential  cross-pollination process ends with particles that are in line with the bearing measurements for every epoch with significantly reduced uncertainty. 

\begin{figure}[tbhp]
  \centering
  \label{fig:SB}\includegraphics[scale=.5]{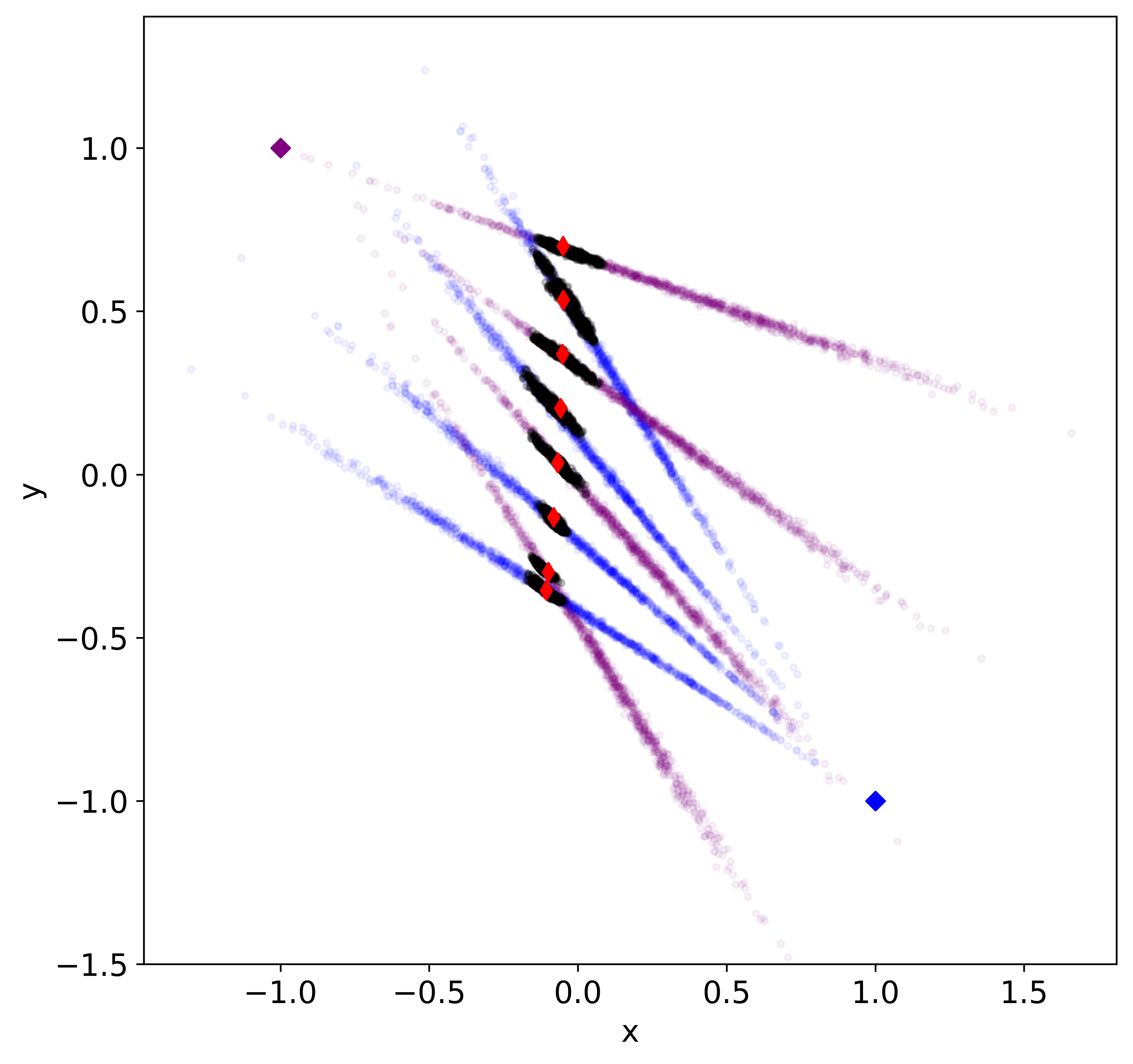}
  \caption{Sequential cross-pollination applied to a two sensor bearing-only tracking problem. Depicted are the sensor locations (purple and blue diamonds), initial particles (purple and blue circles, color matched to the sensor from which the particles have seen data from), cross-pollinated particles (black circles), and true locations (red diamond) for epochs $1,4,7,10,13,16,19, \text{ and } 20$. Time progresses from top to bottom so that epoch $1$ is at the top of the figure and epoch $20$ is at the bottom.}
\end{figure}

\section{Discussion}
\label{sec:End}

We have presented two methods capable of fusing estimates with guarantees in the particle limit. Applying the methods to different examples have led us to some empirical insights about the methods. In our experience, norming-together is more robust to particle degeneracy issues. For example, consider the scenario of fusing two batches of particles from two different sensors where one sensor has far more uncertainty than the other. Applying norming-apart with too few particles can be detrimental in this scenario if the uncertainty from the ``bad'' sensor was so wide that a single particle receives high weight. When pooled together with reasonably weighted particles from the other set this ``best of the bad'' particle would be replicated many times in the resampling framework and could lead to biased estimates. Computations should be performed with sufficient numbers of particles so that this problem is less likely to occur. In this same scenario it may be the case that norming-together only selects particles from the sensor with the reduced uncertainty in the resampling step, while still refining the estimates. This demonstrates that information from more uncertain sensors is still valuable and can be used in improving estimates through cross-pollination. The sequential cross-pollination example demonstrated applicability of cross-pollination ideas to stochastic motion models, however we believe that cross-pollination is most readily applicable in data-starved situations which would correspond to fewer applications of likelihood functions.  

Many ideas must be considered in order to improve the method. Different weighting and resampling schemes for multiple importance sampling lead to improved variances of the estimators~\cite{elvira2017population,elvira2019generalized}. For large sets of data the weights obtained from cross-epoch weighting may render many particles degenerate. Sequential resampling may be used to combat this degeneracy so that fewer likelihoods are considered at a time~\cite{gordon1993novel}. In the batch case, this leads to questions about the order in which the likelihoods should be applied to the weights. For instance, if the order of application of likelihoods is based from the most to least uncertain data improvements may be possible~\cite{pao2000optimal}. 

We are interested in applications and extension of cross-pollination to the areas of multiple target tracking, distributed sensing, and optimal control. In multiple target tracking, track estimates may be cross-pollinated if they are associated to the same object. As both methods pool the sets of particles together, they could be considered centralized fusion methods. Norming-apart is decentralized before the particles are pooled together\textemdash only the likelihoods from the other are required to update the individual sets of particles for an accurate posterior estimates. The ideas in the norming-apart approach could have interesting applications in distributed sensing. Rather than passing state vectors or model parameters from sensor to sensor, likelihood functions could be passed between sensors as a way to refine estimates. This is especially interesting in the case of limited communication bandwidths between sensors. For optimal control, it is possible to identify optimal paths to reach a specified target and associated optimal actions from Monte Carlo sampling of paths under an uncontrolled ({\it i.e.,} open loop) stochastic motion model~\cite{kappen2005path}.  In some application instances, information about candidate paths can be available from different sensors or agents, which might then be combined with cross-pollination resampling to enable optimal control solutions with greater computational efficiency.

\appendix
\section{Theory}
\label{sec:theory}

We now prove that the final resampled measure from the norming-apart cross-pollination process, $\pi^N$, converges to $\pi$ in the particle limit. In particular, we prove this convergence if the initial particles were sourced via importance sampling. We make use of the following ``root mean square" distance, notation, and operators on the set of probability measures presented by Law et al.~\cite[p.87-93]{law2015data}: 
\begin{equation}
d(\mu,\nu) = \sup_{||f||_{\infty}\leq 1} \sqrt{(\mathbb{E}|\mu(f)-\nu(f)|^2)},
\end{equation}
where
\begin{equation}
    \mu(f) = \int_{\mathbb{R}^{n}} f(v)\mu(dv),
\end{equation}
 
\begin{equation}
    \left(L_j\mu\right)(d\nu) =\frac{g_j(\nu)\mu(d\nu)}{\int_{\mathbb{R}^n}g_j(\nu)\mu(d\nu)},
\end{equation}
and
\begin{equation}
    \left(S^N\mu\right)(d\nu) = \frac{1}{N}\sum_{n=1}^N\delta_{\nu^{(n)}}(d\nu), \quad \nu^{(n)} \sim \mu \text{ i.i.d.}.
\end{equation}
Here, $g_j$ is the likelihood function corresponding to the $j$th data point $D_{j}$. We also make use of two lemmas proved by Law et al.~\cite[p.87-93]{law2015data}: 

\begin{lemma}
\label{lem:samp}
The sampling operator satisfies
\begin{equation}
    \sup_{\mu \in P(\mathbb{R}^N)} d(S^N\mu,\mu)<\frac{1}{\sqrt{N}}.
\end{equation}
\end{lemma}

\begin{lemma}
\label{lem:like}
If there exists $\kappa\in(0,1]$ such that for all $\nu\in\mathbb{R}^N$ and $j\in\mathbb{N}$, $\kappa\leq g_j(\nu)\leq \kappa^{-1}$, then
\begin{equation}
\label{eq:likelihoodlemma}
    d(L_j\nu,L_j\mu)\leq 2\kappa^{-2}d(\nu,\mu).
\end{equation}
\end{lemma}

We can now frame the norming-apart process as a sequence of applications of $S^N$ and $L_j$. We will use $L_{-j}$ to be transforming a measure, which has incorporated the $j_{th}$ data, using the Radon-Nikodym derivative to incorporate the remaining data (via a product of likelihood functions) to an appropriate approximate measure of the full posterior.
We can update the empirical measure equations, \eqref{eq:pj0}-\eqref{eq:pn}, using this notation:
\begin{align}
    \pi_{j,0}^N &= S^N L_j S^NP, \\
    \pi_j^N &= L_{-j}\pi_{j,0}^N = L_{-j}S^N L_j S^NP, \\
    \pi^{MN} &= \frac{1}{M}\sum_{j=1}^M \pi^N_j  =\frac{1}{M}\sum_{j=1}^ML_{-j}S^N L_j S^NP, \\
    \pi^{N} &= S^N\pi^{MN},
\end{align}
where $P$ is the prior distribution of $\boldsymbol{\theta}$.

Armed with this theory we can prove our result with some applications of the triangle inequality. We are loose with our application of the likelihood lemma \cref{eq:likelihoodlemma} in our proof summarizing the result as $d(L_j\nu,L_j\mu)\leq C_j d(\nu,\mu)$.

\begin{theorem}
Assumptions: Those of \cref{lem:like} and independent measurements.
Then,
\begin{align}
    d(\pi^N,\pi)< \frac{C_{M}}{\sqrt{N}}
\end{align}
\end{theorem}

\begin{proof}
\begin{align*}
     d(\pi^N, \pi) &\leq d(\pi^N,\pi^{MN}) + d(\pi^{MN},\pi) \\
     &= d(S^N\pi^{MN},\pi^{MN}) + d(\pi^{MN},\pi) \\
     &\leq \frac{1}{\sqrt{N}} + d\left(\frac{1}{M}\sum_{j=1}^M\pi^N_j,\pi\right) \\
     &\leq \frac{1}{\sqrt{N}} + \frac{1}{\sqrt{M}}\sum_{j=1}^M d(\pi^N_j,\pi) \\
     &\leq \frac{1}{\sqrt{N}} + \frac{1}{\sqrt{M}}\sum_{j=1}^M d(L_{-j}\pi^N_{j,0},L_{-j}L_jP) \\
     &\leq \frac{1}{\sqrt{N}} + \frac{1}{\sqrt{M}}\sum_{j=1}^M C_{-j}d(\pi^N_{j,0},L_jP) \\
     &\leq \frac{1}{\sqrt{N}} + \frac{1}{\sqrt{M}}\sum_{j=1}^M C_{-j}d(S^NL_jS^NP,L_jP) \\
     &= \frac{1}{\sqrt{N}} + \frac{1}{\sqrt{M}}\sum_{j=1}^M C_{-j}d(S^NL_jS^NP,L_jP) \\
     &\leq \frac{1}{\sqrt{N}} + \frac{1}{\sqrt{M}}\sum_{j=1}^M C_{-j}(d(S^NL_jS^NP,L_jS^NP)+d(L_jS^NP,L_jP)) \\
     &\leq \frac{1}{\sqrt{N}} + \frac{1}{\sqrt{M}}\sum_{j=1}^M C_{-j}\left(\frac{1}{\sqrt{N}}+C_jd(S^NP,P)\right) \\
     &\leq \frac{1}{\sqrt{N}} + \frac{1}{\sqrt{M}}\sum_{j=1}^M C_{-j}\left(\frac{1}{\sqrt{N}}+C_j\frac{1}{\sqrt{N}}\right) \\
     & = \frac{1+\frac{1}{\sqrt{M}}\sum_{j=1}^M(C_{-j}+C_{-j}C_j)}{\sqrt{N}}
\end{align*}
\end{proof}

We have shown convergence of the norming-apart method in the particle limit. The fourth line of the proof is not trivial, it was achieved by bounding cross terms of the squared sum by a sum of the two terms squared, that is, $ab\leq (a^2+b^2)/2$ followed by simple bounding arguments. By noting that the norming-together method can be written to look like the norming-apart method,

\begin{equation}
    \pi_{t}^{MN} = \sum_{j=1}^M K_j \pi_j^N, \quad K_j=\frac{\sum_{i=1}^N w_{j}^{(i)}}{\sum_{j'=1}^M \sum_{i=1}^N w_{j'}^{(i)}},
\end{equation}
we may employ a similar proof to show that the norming-together method converges in the particle limit.

\section{Pseudo-code}
\label{sec:code}

\begin{algorithm}
\caption{Norming-Apart Cross-Pollination}
\label{alg:NACP}
\begin{algorithmic}
\STATE \textbf{Input:} Acquire the initial samples that comprise $\pi_{j,0}^N$.

\begin{ALC@g}
    \FOR{$j=1$ \TO $M$}
        \STATE Obtain $\boldsymbol{\theta}_j^{(i)}\,(i=1,\ldots,N)$ 
    \ENDFOR
\end{ALC@g}

\STATE \textbf{Cross-Epoch Weighting:} Form $\pi_j^N$.
\begin{ALC@g}
\FOR{$j=1$ \TO $M$}
        \STATE Calculate weights $w_j^{(i)}\,(i=1,\ldots,N)$ according to equation ~\eqref{eq:pjn}.
\ENDFOR
\end{ALC@g}

\STATE \textbf{Pool Particles and Weights:} Form $\pi^{MN}$.
\begin{ALC@g}
\STATE Form $T = \bigcup_{j=1}^M \bigcup_{i=1}^N  \boldsymbol{\theta}^{(i)}_j$ and $W = \bigcup_{j=1}^M \bigcup_{i=1}^N  \left(\frac{1}{M} w_j^{(i)}\right)$.
\end{ALC@g}

\STATE \textbf{Resample:} Form $\pi^N$.
\begin{ALC@g}
\STATE Perform multinomial sampling on the set of weights, $W$, to obtain $N$ (from $M\times N$ total) particles $\boldsymbol{\theta}^{(i)}\,(i=1,\ldots,N)$ from $T$.
\end{ALC@g}

\end{algorithmic}
\end{algorithm}

\begin{algorithm}
\caption{Norming-Together Cross-Pollination}
\label{alg:NTCP}
\begin{algorithmic}
\STATE \textbf{Input:} Acquire the initial samples that comprise $\pi_{j,0}^N$.

\begin{ALC@g}
    \FOR{$j=1$ \TO $M$}
        \STATE Obtain $\boldsymbol{\theta}_j^{(i)}\,(i=1,\ldots,N)$ 
    \ENDFOR
\end{ALC@g}

\STATE \textbf{Cross-Epoch Weighting:} Form $\pi_j^N$.
\begin{ALC@g}
\FOR{$j=1$ \TO $M$}
        \STATE Calculate weights $w_j^{(i)}\,(i=1,\ldots,N)$ according to equation ~\eqref{eq:pmnt}.
\ENDFOR
\end{ALC@g}

\STATE \textbf{Pool Particles and Weights:} Form $\pi_t^{MN}$.
\begin{ALC@g}
\STATE Form $T = \bigcup_{j=1}^M \bigcup_{i=1}^N  \boldsymbol{\theta}^{(i)}_j$ and $W = \bigcup_{j=1}^M \bigcup_{i=1}^N w_j^{(i)}$.
\end{ALC@g}

\STATE \textbf{Resample:} Form $\pi^N$.
\begin{ALC@g}
\STATE Perform multinomial sampling on the set of weights, $W$, to obtain $N$ (from $M\times N$ total) particles $\boldsymbol{\theta}^{(i)}\,(i=1,\ldots,N)$ from $T$.
\end{ALC@g}

\end{algorithmic}
\end{algorithm}

\begin{algorithm}[H]
\caption{Sequential Cross-Pollination}
\label{alg:SCP}
\begin{algorithmic}

\STATE \textbf{Initialization:}
\begin{ALC@g}
    \STATE Obtain $\boldsymbol{\theta}_1^{(i)}\,(i=1,\ldots,N)$  that comprise $\pi^N_{1,0}$ corresponding to $P(\boldsymbol{\theta}_1\mid z_1)$
    \STATE Set $\boldsymbol{\theta}_{1,*}^{(i)} = \boldsymbol{\theta}_{1}^{(i)}$
\end{ALC@g}

\STATE \textbf{Iteration Cross-Pollination:}  

\begin{ALC@g}
    \FOR{j=2 \TO M}
        \STATE Obtain $\boldsymbol{\theta}_j^{(i)}\,(i=1,\ldots,N)$  that comprise $\pi^N_{j,0}$ corresponding to $P(\boldsymbol{\theta}_j\mid z_j)$
        \STATE Compute unnormalized weights $\widetilde{w}_{j}^{(i)} = \prod_{k< j}g_k\left( \boldsymbol{\theta}_j^{(i)}\right)$
        \STATE Motion $\boldsymbol{\theta}_{j-1,*}^{(i)}$ forward to $t_j$ and set $w_{j,*}^{(i)} = g_j\left(\boldsymbol{\theta}_{j-1,*}^{(i)}\right)$
        \STATE Normalize weights, pool particles and weights, and resample particles as in \ref{alg:NACP} or \ref{alg:NTCP}
        \STATE Set result equal to $\boldsymbol{\theta}_{j,*}^{(i)}$ 
    \ENDFOR
\end{ALC@g}
\end{algorithmic}
\end{algorithm}

\section*{Acknowledgements}
The authors thank Ian Grooms for valuable discussions.

\bibliographystyle{siamplain}
\bibliography{bib}

\end{document}